\documentclass[11pt]{article}

\usepackage{charter}
\usepackage[charter]{mathdesign}
\usepackage{microtype}

\usepackage[margin=1in,letterpaper]{geometry}

\usepackage{setspace}
\usepackage{enumitem}

\usepackage{cite}

\usepackage{graphicx}
\usepackage{subfigure}

\usepackage{booktabs}

\usepackage{amsthm}

\newtheorem{lem}{Lemma}

\newcommand{\newparentheses}[3]{%
  \expandafter\newcommand\csname #1\endcsname[1]{#2##1#3}%
  \expandafter\newcommand\csname #1L\endcsname[1]{\bigl#2##1\bigr#3}%
  \expandafter\newcommand\csname #1XL\endcsname[1]{\Bigl#2##1\Bigr#3}%
  \expandafter\newcommand\csname #1V\endcsname[1]{\left#2##1\right#3}}

\newparentheses{parens}{(}{)}
\newparentheses{braces}{\{}{\}}
\newparentheses{brackets}{[}{]}
\newparentheses{floor}{\lfloor}{\rfloor}
\newparentheses{ceil}{\lceil}{\rceil}
\newparentheses{abs}{|}{|}
\newparentheses{set}{\{}{\}}
\newparentheses{size}{|}{|}
\newparentheses{seq}{\langle}{\rangle}

\makeatletter
\newcommand{\onenewattribute}[3]{%
  \@ifundefined{#1}{\let\@@def\newcommand}{\let\@@def\renewcommand}%
  \expandafter\@@def\csname #1\endcsname[2][]{%
    \ifthenelse{\equal{##1}{}}%
    {#2\csname #3\endcsname{##2}}%
    {#2_{##1}\csname #3\endcsname{##2}}}}
\newcommand{\newattribute}[2]{%
  \onenewattribute{#1}{#2}{parens}%
  \onenewattribute{#1L}{#2}{parensL}%
  \onenewattribute{#1XL}{#2}{parensXL}%
  \onenewattribute{#1V}{#2}{parensV}}
\makeatother

\newattribute{OhOf}{\mathrm{O}}
\newattribute{ThetaOf}{\Theta}
\newattribute{OmegaOf}{\Omega}
\newattribute{ohOf}{\mathrm{o}}
\newattribute{omegaOf}{\omega}

\newattribute{poly}{\mathrm{poly}}

\newcommand{\C}{\mathcal{C}}
\newcommand{\E}{\mathcal{E}}
\newcommand{\F}{\mathcal{F}}
\newcommand{\G}{\mathcal{G}}
\renewcommand{\H}{\mathcal{H}}

\newcommand{\T}{\mathcal{T}}

\usepackage{color}

\definecolor{green}{rgb}{0,0.5,0}
\definecolor{magenta}{rgb}{0.75,0,0.75}
\definecolor{brown}{rgb}{0.75,0.25,0}

\newcommand{\commentmarkerfont}{\tiny}
\newcommand{\commenttextfont}{\scriptsize}

\newcommand{\leo}{}
\newcommand{\revA}{}

\newcounter{comment}
\newcommand{\comment}[2]{}

\newcommand{\makecommentmarker}[1]{%
  \begingroup
  \setlength{\fboxsep}{1pt}%
  \csname #1commentcolor\endcsname
  \fbox{\commentmarkerfont\thecomment}%
  \endgroup}
\newcommand{\commenttext}[2]{%
  \marginpar{%
    \raggedright
    \leavevmode
    \csname #1commentcolor\endcsname
    \makecommentmarker{#1}
    \begingroup\setlength{\fboxsep}{1pt}\fbox{\commentmarkerfont #1}\endgroup\
    \commenttextfont #2}}

\begin{document}

\title{Hybridization Number on Three \revA{Rooted Binary} Trees \revA{is EPT}}
\author{Leo van Iersel
\footnote{Delft Institute of Applied Mathematics, Delft University of Technology, P.O. Box 5, 2600 AA Delft, The Netherlands, l.j.j.v.iersel@gmail.com}
\thanks{Leo van Iersel and Nela Leki\'{c} were respectively funded by a Veni and a Vrije Competitie grant from The Netherlands Organisation for Scientific Research (NWO). Leo van Iersel was partly funded by the 3TU Applied Mathematics Institute.}
\and Steven Kelk
\footnote{Department of Knowledge Engineering (DKE), Maastricht University, P.O. Box 616, 6200 MD Maastricht, The Netherlands, steven.kelk@maastrichtuniversity.nl, nela.lekic@maastrichtuniversity.nl}
\and Nela Leki\'{c}
\footnotemark[2]
\footnotemark[3]
\and Chris Whidden
\footnote{Fred Hutchinson Cancer Research Center,
1100 Fairview Ave. N.,
PO Box 19024,
Seattle, WA, USA 98109,
whidden@cs.dal.ca.
Chris Whidden is a Simons Foundation Fellow of the Life Sciences Research Foundation and was supported by National Science Foundation award 1223057.}
\and Norbert Zeh
\footnote{Faculty of Computer Science,
Dalhousie University,
6050 University Ave,
Halifax, NS B3H 1W5,
Canada,
nzeh@cs.dal.ca.
Research funded in part by the Natural Sciences and Engineering Research Council of Canada and the Canada Research Chairs programme.}
}
\maketitle


\begin{abstract}\noindent
Phylogenetic networks are leaf-labelled directed acyclic graphs that are used to describe non-treelike evolutionary histories and are thus a generalization of phylogenetic trees. The hybridization number of a phylogenetic network is the sum of all in-degrees minus the number of nodes plus one. The \textsc{Hybridization Number} problem takes as input a collection of \revA{rooted binary} phylogenetic trees and asks to construct a phylogenetic network that contains an embedding of each of the input trees and has the smallest possible hybridization number. We present an algorithm for the \textsc{Hybridization Number} problem on three binary phylogenetic trees on~$n$ leaves that runs in time $\OhOf{c^k \poly{n}}$, with~$k$ the hybridization number of an optimal network and~$c$ some \revA{(astronomical)} constant. For the case of two trees, an algorithm with running time $\OhOf{3.18^k n}$ was proposed before whereas an algorithm with running time $\OhOf{c^k \poly{n}}$, \revA{also called an EPT algorithm, had prior to this article remained elusive for more than two trees}. The algorithm for two trees uses the close connection to acyclic agreement forests to achieve a linear exponent in the running time, while previous algorithms for more than two trees (explicitly or implicitly) relied on a brute force search through all possible underlying network topologies, leading to running times that are not $\OhOf{c^k \poly{n}}$, for any~$c$. The connection to acyclic agreement forests is much weaker for more than two trees, so even given the right agreement forest, the reconstruction of the network poses major challenges. We prove novel structural results that allow us to reconstruct a network without having to guess the underlying topology. Our techniques generalize to more than three input trees with the exception of one key lemma that maps nodes in the network to tree nodes in order to minimize the amount of guessing involved in constructing the network. The main open problem therefore is to prove results that establish such a mapping for more than three trees.
\end{abstract}

\clearpage

\section{Introduction}

In computational biology the evolutionary history of a set of contemporary species (or \emph{taxa}) is often modelled
as a \emph{rooted phylogenetic tree}. Informally this is a rooted tree in which the leaves are bijectively labelled by
the taxa and edges are directed away from the root, reflecting the direction of evolution \cite{SempleSteel2000}. Nodes of
out-degree two or higher model the points in history at which a common ancestor of a subset of the taxa differentiated
into two or more sublineages. The central problem in phylogenetics is to recover the topology of the ``true''
phylogenetic tree, given only information about the taxa, often DNA data. This is a challenging computational
problem and has been the topic of intensive research during the last 40 years \cite{reconstructingevolution}. Recently our understanding of
evolutionary mechanisms has deepened and there is growing awareness that evolution is not always treelike \cite{expanding}. In particular,
due to \emph{reticulate phenomena} such as hybridization and horizontal gene transfer \cite{davidbook}, the evolution of a set of
species is sometimes better modelled as a \emph{rooted phylogenetic network}~\cite{HusonRuppScornavacca10}, essentially a generalization of
phylogenetic trees to directed acyclic graphs (DAGs). In such graphs, nodes with indegree two or higher, known as \emph{reticulations}, represent the points at which two or more lineages merge, rather than diversify.

The study of rooted phylogenetic networks is comparatively new and has given rise to many novel and hard combinatorial optimization problems \cite{HusonRuppScornavacca10}. In this article we focus on the
\textsc{Hybridization Number} problem, originally introduced in \cite{BaroniEtAl2005,BaroniEtAl2004}, which is one of the most well-studied phylogenetic network problems to date. Here
we are given a set of rooted phylogenetic trees~$\mathcal{T}$, on the same set of taxa~$X$, and the goal is to
construct a phylogenetic network---henceforth called a \emph{hybridization network}---that contains an image
of each of the input trees, while minimizing the hybridization number~$k$ of the network. If we restrict (without loss of generality) to networks with maximum in-degree two, the hybridization number is simply equal to the number of reticulation nodes. We defer exact definitions to the preliminaries. \leo{See Figure~\ref{fig:invisible} for an example of a hybridization network (with hybridization number three) for three input trees.}

The holy grail for this problem is to develop algorithms that can cope with
many input trees and non-binary input trees~\cite{davidbook} (and to take different causes of incongruence into account, see e.g. \cite{yu2013parsimonious}).
 However, thus far most algorithmic research has focused on the simplest possible case: $|\mathcal{T}|=2$ and both input trees are binary. Unfortunately even this version of the problem is NP-hard and APX-hard \cite{bordewich}, with similar
(in)approximability properties to the classical problem \textsc{Directed Feedback Vertex Set}~\cite{cyclekiller}. Fortunately the binary
two-tree problem is fixed parameter tractable (FPT) in $k$. (See \cite{downey1999,niedermeier2006} for
an introduction to fixed parameter tractability). This result was initially established via kernelization---the problem has a quadratic kernel \cite{sempbordfpt2007}---but the theoretical state of the art is an algorithm based on bounded-search with running time $\OhOf{ 3.18^{k} \cdot \poly{n} }$ \cite{whidden2013fixed}, where~$n=|X|$. The comparative tractability of the problem, both
in theory and practice (see e.g. \cite{chen2013ultrafast} for a fast implementation), stems from the essentially one-to-one relationship between solutions to the two-tree problem and the \emph{Maximum Acyclic Agreement Forest} (MAAF) problem. In the latter problem (originally introduced in \cite{BaroniEtAl2005}) one is required to cut the two input trees into common components so that the number of components is minimized and there are no cyclical dependencies between components. The MAAF abstraction gives a useful static characterization of the two-tree \textsc{Hybridization Number} problem \cite{bordewich}. In particular, in the two-tree case the MAAF abstraction essentially allows us to bypass the problem
of actually constructing the hybridization network: it can easily be constructed  in polynomial time from the components
of the MAAF.  The MAAF abstraction, and related FPT results, also hold in the case of two \emph{non}-binary trees, albeit with significant technical complications \cite{linzsemple2009,simplefpt}. 

For $|\mathcal{T}| > 2$ the situation becomes more complex, however, even when restricted to binary trees\footnote{For the rest of the introduction we focus only on the case of binary trees---see \cite{WG2014} and \cite{nonbinary} for an overview of
recent non-binary results. The non-binary case is a generalization of the binary case and therefore inherits all
the negative results, but not the positive results, of the binary problem.} and $|\mathcal{T}|=3$. The MAAF abstraction weakens significantly and cannot (obviously) be used to generate optimal solutions to the \textsc{Hybridization Number} problem.  Without the MAAF abstraction it seems that we have to explicitly confront the
challenge of actually constructing the hybridization network itself. This is a theoretically daunting challenge, since the
space of DAGs is huge. The good news is that for $|\mathcal{T}|>2$ the problem nevertheless remains FPT in $k$ \cite{vanIerselLinz,towards}. The
bad news is that none of these results satisfactorily address the problem of actually constructing
the network. The FPT result in \cite{vanIerselLinz} gives a quadratic kernel, but does not describe a (good) algorithm for solving the kernel. The bounded-search FPT result in \cite{towards}, based on \cite{KelkScornavacca2011}, does actually construct the network, but has an astronomical running time. The running time is so large because it brute forces over the space of all possible \emph{generators}, i.e., possible ``backbone topologies'' of the network~\cite{elusiveness,KelkScornavacca2011}, a space which is not known to be~$\OhOf{c^k}$, and continues with a tower of guesses, which is not~$\OhOf{c^k}$, for each such generator. At present, therefore, the only FPT algorithms for the
case of three binary trees are either kernelizations, or bounded-search algorithms with an exponential dependency on~$k$ with a non-linear exponent. Several exponential-time algorithms do exist, such as \cite{wu2013algorithm} and the algorithm discussed in \cite{vanIerselLinz}, but using them to solve a kernelized hybridization number instance unfortunately does not help for two reasons. Firstly, the size of the best-known kernel (i.e. the number~$n$ of leaves of a kernelized instance) is quadratic, and not linear in~$k$. Secondly, no previously-known exponential-time algorithm has an~$\OhOf{c^{n}}$ running time. Therefore, the challenge is to determine whether an algorithm with running time $\OhOf{ c^{k} \cdot \poly{n} }$ exists for the case of three binary trees. In other words, is the problem EPT~\cite{flum2006bounded}?

  \begin{figure}[p]
    \centering
    \includegraphics{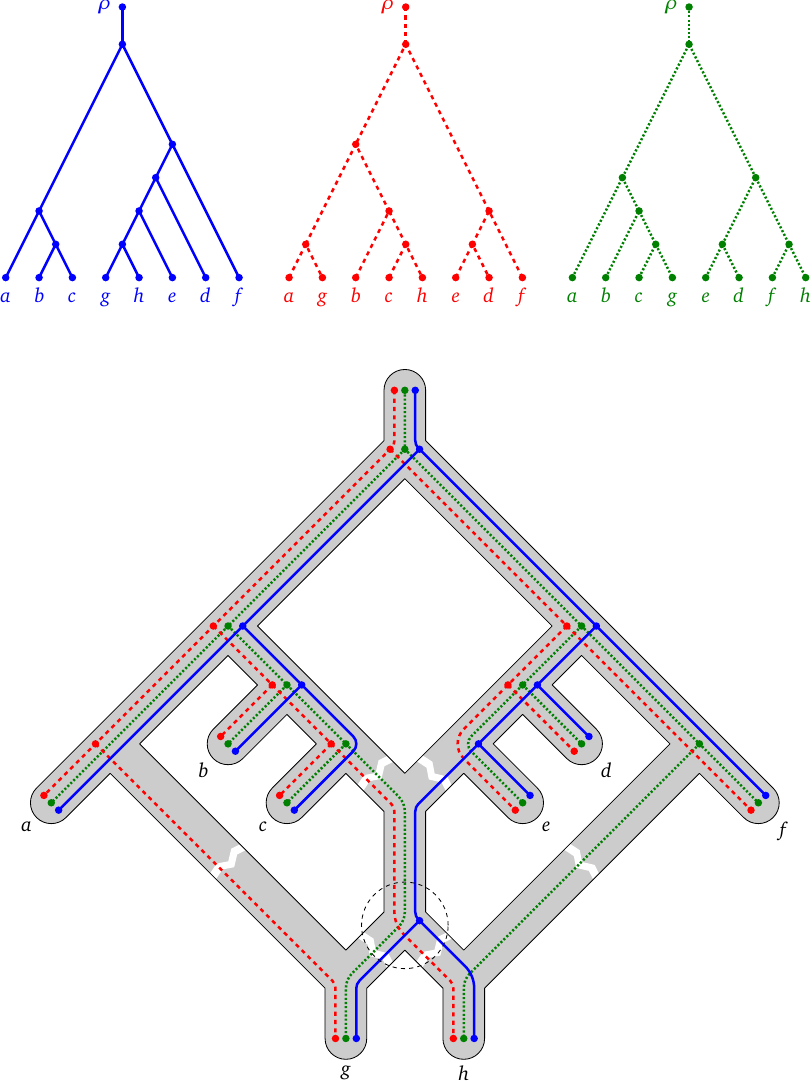}
    \caption{A hybridization network for three trees which contains an invisible component (inside the dashed circle). It can be shown that any hybridization network for these trees contains an invisible component. \leo{However, the single node inside this component can be identified beause it corresponds to a node of the blue solid input tree.}}
    \label{fig:invisible}
  \end{figure}

In this article we answer this challenge positively. Although the constant $c$ that we find is astronomical---1,609,891,840---it represents a significant development in our understanding of the underlying combinatorial structure of the \textsc{Hybridization Number} problem. We show that, although it is not clear how a MAAF can be pieced together into an optimal solution to the \textsc{Hybridization Number} problem, it is still possible to identify in $\OhOf{c^k \poly{n}}$ time a (not necessarily maximum) acyclic agreement forest that does have this property. Having found the appropriate acyclic agreement forest, we use deep\comment{LI}{novel?} insights into the structure of optimal hybridization networks to piece the components of the forest together into a network. \leo{The difficulty of this step comes from the fact that, unlike in the two-tree case, it is no longer possible to avoid having nodes in the network that are separated from all leaves by hybridization edges, and that are hence not represented in the agreement forest. The main insight helping to overcome this problem is that,} in the case $|\mathcal{T}|=3$, there always exists an optimal hybridization network such that each of its out-degree-2 nodes corresponds to nodes of one or more of the input trees, \leo{see Figure~\ref{fig:invisible}}. This enables us to keep the combinatorial explosion in the number of possible network topologies under control. 

Note that our algorithm can be
viewed as a structural generalization of existing algorithms for two trees, which also separate
the identification of the underlying acyclic agreement forest and the construction of the network
into two phases. In the case of two trees
the second phase is polynomial and it is comparatively easy to obtain $\OhOf{c^k \poly{n}}$ running times for the first phase.
In fact, although our overall result at present only holds for the case $|\mathcal{T}|=3$, the
results for the first phase hold without modification
for the case $|\mathcal{T}| > 3$. As we demonstrate, the only barrier to extending our result is the fact that,
for $|\mathcal{T}|>3$, the combinatorial insight mentioned in the previous paragraph no longer holds. Indeed, there are two new challenges stemming from this article. Firstly, to adapt and generalize the combinatorial insight so that the wider result can be extended to four or more trees. Secondly, to significantly optimize the constant $c$ in our running time. How close can we get to the competitive $\OhOf{ 3.18^{k} \cdot \poly{n} }$ running time achieved in the case of two trees?

The structure of the remainder of this article is as follows. In Section~\ref{sec:prelim}, we present the necessary definitions. Section~\ref{sec:aaf-guessing} shows how we can guess the underlying acyclic agreement forest of an optimal hybridization network in $\OhOf{c^k \poly{n}}$ time. Then we define a notion of \revA{``tight''} networks in Section~\ref{sec:network-to-tight-network} (basically, networks where each out-degree-2 node corresponds to a node of at least one of the input trees) and show that we may restrict our attention to \revA{tight} networks as long as there are at most three trees in the input. Subsequently, Section~\ref{sec:guessing-things} shows how such a \revA{tight} network can be reconstructed from an acyclic agreement forest and the input trees in $\OhOf{c^k \poly{n}}$ time. Finally, we present our conclusions in Section~\ref{sec:concl} and give an example of the algorithm in the appendix.

\section{Preliminaries}\label{sec:prelim}

A \emph{rooted phylogenetic $X$-tree} is a rooted tree with no nodes with in-degree~1 and out-degree~1, a \revA{node} with in-degree 0 and out-degree 1 \revA{(the root)}, and leaves bijectively labelled with the elements of a finite set~$X$. Such a tree is called binary if all inner nodes except the root have in-degree~1 and out-degree~2. From now on we will refer to a rooted binary phylogenetic $X$-tree as a \emph{tree} for short, \revA{since we only consider rooted binary trees that are all on the same set~$X$}. The convention that roots have out-degree 1 is not essential but for technical convenience.

A \emph{rooted phylogenetic network} \revA{(on~$X$)} is a directed acyclic graph (DAG) with no nodes with in-degree 1 and out-degree 1, a single in-degree-$0$ node (the \emph{root}) with out-degree 1 and leaves bijectively labelled with the elements of~$X$. Rooted phylogenetic networks will be called \emph{networks} for short. We identify each leaf of a tree or network with its label and call directed edges \emph{edges} for short. \leo{We see the root of a tree or network as a leaf and assume without loss of generality that it is labelled~$\rho$.} We call a network \emph{binary} if every \leo{non-leaf} node has total degree~3 and all leaves have degree~1.

We call network nodes with in-degree~1 and out-degree \leo{at least~2} \emph{split nodes}, while nodes with in-degree \leo{at least~2} are called \emph{reticulation nodes}, or \emph{reticulations} for short. The \emph{hybridization number} (often also called reticulation number) of a binary network~$N$ is defined as the number of reticulation nodes of~$N$. For a general network the hybridization number is given by the sum $\sum(d^-(v)-1)$ over all nodes~$v$ \leo{of}~$N$ with in-degree $d^-(v)$ at least 2. For a tree $T$ and a set $X' \subseteq X$, we define $T(X')$ as the minimal subtree of $T$ that contains all elements of $X'$, and \revA{$T|X'$} as the result of suppressing all nodes of $T(X')$ with in- and out-degree~1. The set of leaves of a tree~$T$ is denoted $L(T)$.

We say that a tree~$T$ is \emph{displayed} by a network~$N$ if~$T$ can be obtained from a subgraph of~$N$ by contracting edges. Given a set~$\T$ of rooted phylogenetic trees, the \textsc{Minimum Hybridization} problem asks to find a phylogenetic network~$N$ that displays each tree in~$\T$ such that the \leo{hybridization} number of~$N$ is minimized. We say that~$N$ is a \emph{hybridization network for} a set~$\T$ of input trees if~$N$ displays all~$T\in\T$. In addition, we say that the \emph{hybridization number} of a set of input trees~$\T$ is the hybridization number of a hybridization network for~$\T$ \leo{that has the lowest hybridization number over all hybridization networks for~$\T$}. It is well known, and easy to see, that if there exists a hybridization network for~$\T$, there also exists a binary hybridization network for~$\T$ \leo{with the same hybridization number} (see e.g.~\cite{towards}). Therefore, all hybridization networks are from now on assumed to be binary.

Let~$T$ be a rooted, \revA{binary} phylogenetic $X$-tree and~$S$ a rooted, \revA{binary} phylogenetic $X'$-tree for some $X' \subseteq X$. We say that~$S$ is a \emph{pendant subtree} of~$T$ if it is a subtree that can be detached from~$T$ by deleting a single edge. For a set~$\T$ of phylogenetic $X$-trees and $X' \subseteq X$, a \emph{common pendant subtree} of~$\T$ is a rooted phylogenetic $X'$-tree that is a pendant subtree of each tree in~$\T$. \revA{A common pendant subtree is called \emph{trivial} if it consists of a single leaf.} Let~$T$ be a tree, let $(x_1, x_2,...,x_q)$ be a tuple of elements of~$X$ with \leo{$q \geq 1$} and let $p_i$ denote the parent of~$x_i$ in~$T$. We say that the tuple $(x_1, x_2,...,x_q)$ is a \emph{chain} of~$T$ if either $(p_q, p_{q-1}, \ldots, p_1 )$ is a directed path in $T$, or $(p_q, p_{q-1}, ..., p_2)$ is a directed path in $T$ and $p_1 = p_2$. A \emph{common chain} of a set~$\T$ of trees is a maximal tuple $(x_1, x_2,...,x_q)$ that is a chain of each tree in~$\T$.


Related to the hybridization number problem is a concept of agreement forests. A forest is a collection of trees, which we will call components rather than trees to avoid confusion with the input trees. We say that a forest~$\F$ is a \emph{forest for} a tree~$T$ if $T|L(F)$ is isomorphic\comment{LI}{It is probably clear that we need that each leaf is mapped to the leaf with the same label} to~$F$ for all~$F\in\F$ and the trees $\{T(L(F))\mid F \in \F \}$ are node-disjoint subtrees of~$T$ whose leaf-set union equals $L(T)$. By this definition, if~$\F$ is a forest for some tree~$T$, then $\{L(F)\mid F \in \F\}$ is a partition of the leaf set of~$T$. It will indeed sometimes be useful to see a forest as a partition of the leaves and sometimes to see it as a collection of trees. If $\T$ is a set of trees, then a forest~$\F$ is an \emph{agreement forest} of~$\T$ if it is a forest for each~$T\in\T$. Note that these definitions only apply to binary trees.

We define the \emph{inheritance graph} $\mathit{IG}(\T, \F)$ of an agreement forest~$\F$ of a set~$\T$ of trees as the directed graph whose node set is the set of components of~$\F$ and whose edge set contains an edge~$(F,F')$ precisely if there is a directed path from the root of $T(L(F))$ to the root of $T(L(F'))$ in at least one tree $T\in\T$. An agreement forest~$\F$ of $\T$ is called an \emph{acyclic agreement forest} (AAF) if the graph $\mathit{IG}(\T,\F)$ does not contain any directed cycles.

The last definition we need is the notion of generators (see e.g.~\cite{KelkScornavacca2011}), which we use to describe the underlying structure of networks without non-trivial pendant subtrees. A (binary) \emph{$r$-reticulation generator} is defined as an acyclic directed multigraph with a single root with in-degree 0 and out-degree~1 and exactly~$r$ nodes with in-degree~2 and out-degree at most~1; all other nodes have in-degree 1 and out-degree 2. If~$N$ is a network, then the \emph{underlying generator} of~$N$ is the generator obtained from~$N$ by deleting all leaves and suppressing all in-degree-1 out-degree-1 nodes. The \emph{sides} of a generator are its edges (the \emph{edge sides}) and its nodes with in-degree 2 and outdegree 0 (the \emph{node sides}). Thus, each leaf of a network~$N$ is on a certain side of its underlying generator. See Figure~\ref{fig:generator} for an example.

\begin{figure}[t]
 \centering
  \includegraphics{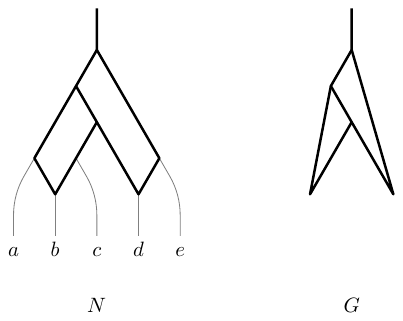}
 \caption{The graph $G$ is the generator of the network $N$. It has 9 sides: 7 edge sides and 2 node sides. The leaves~$b$ and~$d$ are on node sides while~$a$,  $c$, and~$e$ are on edge sides.}
 \label{fig:generator}
\end{figure}

\section{Guessing the AAF}

\label{sec:aaf-guessing}

\noindent
\leo{Let~$\T$ be a collection of input trees.} Without loss of generality we will assume that \leo{$\T$ contains} no non-trivial common pendant subtrees (because each such subtree can be replaced by a single leaf). \leo{In this section, we show how we can guess an AAF from which we can build an optimal hybridization network for~$\T$. To make this precise, we define the \emph{deletion forest} of a network~$N$ as} the forest obtained from~$N$ by deleting all the edges entering reticulation nodes, deleting all resulting connected components that do not contain any taxa, and then taking the partition of the taxa induced by the remaining connected components. Note that for a given network the deletion forest is uniquely defined. \leo{We start by proving the following lemma.}

\begin{lem}
  \label{lem:network-to-aaf}
Given a hybridization network~$N$ with hybridization number~$k$ \leo{for a set~$\T$ of input trees}, the deletion forest of~$N$ \leo{is} an AAF of~$\T$ with at most $k+1$ components.
\end{lem}
\begin{proof}
We first show that the deletion forest contains at most $k+1$ components. To see this, note that~$N$ contains exactly~$k$ reticulation nodes. For a reticulation node~$r$, let $X(r)$ be the set of taxa that can be reached from~$r$ by directed paths that start at~$r$ and which do not intersect with any reticulation \leo{apart from~$r$}. (Possibly, $X(r)=\emptyset$.) By construction, none of the edges on these directed paths are deleted when the deletion forest is created. Hence, all the taxa in $X(r)$ will be in the same connected component. Similarly,~\leo{if $X(\rho)$ denotes the set of taxa reachable by directed paths that start at the root and which do not intersect with any reticulations, then the taxa in~$X(\rho)$ will also be together in the same connected component. Note that the deletion forest~$\F$ of~$N$ (seeing it as a partition of the taxa) is the collection containing~$X(\rho)$ if~$X(\rho)\neq\emptyset$ and~$X(r)$ for each reticulation~$r$ for which~$X(r)\neq\emptyset$. Hence, the deletion forest contains at most~$k+1$ components. Moreover, for each~$F\in\F$, each input tree must yield the same subtree when restricted to the subset of taxa of~$F$ because~$N$ displays all the input trees. In addition, for each input tree~$T\in\T$, the subtrees $\{T(L(F)) \mid F\in\F\}$ are node-disjoint, again because~$N$ displays~$T$. It follows that the deletion forest~$\F$ of~$N$ is indeed an \revA{agreement forest} of the input trees,
with at most $k+1$ components. \revA{Moreover, it is clearly acyclic since the network is acyclic.}}
\end{proof}

As a consequence of Lemma~\ref{lem:network-to-aaf}, we will from now on refer to the deletion forest of a network as its \emph{deletion AAF}. \leo{Next we show how to guess the deletion AAF of some optimal hybridization network for the input trees.} \revA{More precisely, we show how to construct a set of AAFs containing at least one AAF with this property. In Sections~\ref{sec:network-to-tight-network} and~\ref{sec:guessing-things} we will show how to determine from which AAF(s) in the set we can build an optimal hybridization network.}
\begin{lem}
  \label{lem:aaf-to-network}
Let~$k$ be the hybridization number of the \leo{set~$\T$} of input trees. Then, in time
$\OhOf{c^k \cdot \poly{n}}$, we can find a set of AAFs containing at least one deletion AAF of some hybridization network for~$\T$ with hybridization number~$k$.
\end{lem}
\begin{proof}
\leo{Consider an arbitrary input tree~$T\in\T$. Observe that an AAF of~$\T$ with $k'+1$ components can be obtained from~$T$ by deleting exactly~$k'$ edges and taking the partition of the taxa induced by the resulting connected components. By Lemma~\ref{lem:network-to-aaf}, the deletion AAF of an optimal hybridization network for~$\T$ has at most~$k+1$ components. The goal therefore is to locate the at most~$k$ edges that need to be deleted from~$T$ in order to obtain the deletion AAF of some optimal hybridization network for~$\T$.}

\leo{Let~$N$ be a hybridization network for~$\T$ with hybridization number~$k$ and consider its underlying generator, which is a $k$-reticulation generator and hence has at most~$k$ node sides and at most~$4k-1$ edge sides~\cite{vanIerselLinz}. It follows that there are at most $5k-1$ common chains of~$\T$, because any two taxa on the same edge side of~$N$ are in the same common chain~\cite{vanIerselLinz}. The} set of common chains is unambiguously defined by the set of input trees, can be computed in polynomial time, and no two chains can share a taxon. \leo{Moreover,} in~\cite{vanIerselLinz} it is proven that if two or more taxa of a common chain are on a single edge side of the underlying generator, the entire chain can safely be moved onto that edge side. That is, the new network still displays the input trees and has hybridization number no higher than the old network.

This means that we can assume the existence of an optimal network $N'$ such that for each common chain there are exactly two possibilities: (1) the chain is on a single \leo{side} of the underlying generator, or (2) each taxon of the chain is on a different side of the underlying generator. For each chain we can guess which of the two cases holds, using at most $2^{5k-1}$ guesses for the entire set of chains. \leo{Since, as mentioned before, any two taxa that are on the same side belong to a common chain, it follows that each side of (the underlying generator of)~$N'$ contains a complete case-1 chain, a single taxon (which is either in a case-2 chain or a singleton-chain) or no taxa at all.}

Now, assume that we have identified the correct set of guesses describing the behaviour of the common chains in~$N'$. It remains to show that we can identify the correct set of edges to delete in~$T$ to obtain the deletion AAF corresponding to $N'$. Observe that for each case-1 chain it is not necessary to delete any of the internal edges of the chain in~$T$. This is because we have correctly identified that the entire chain is attached to a single edge side of the generator and thus that it belongs to a single component of the deletion AAF. For each of the other edges in~$T$ we simply guess whether to delete it or not. Fortunately, there are not too many of these edges. Specifically, recall that each side contains either a case-1 chain or a single taxon, and that the number of sides is at most $5k-1$. Hence, if we collapse each case-1 chain~$C$ into a new taxon $x_{C}$, which is permitted because we will never cut its internal edges, there are in total at most $5k - 1$ taxa left. A binary tree with $5k-1$ taxa has $10k - 4$ edges. By guessing for each of these edges whether or not to delete it, we observe that, in total, the deletion AAF of~$N'$ can be located in time at most $\OhOf{ 2^{5k} \cdot 2^{10k} \cdot \poly{n} }$.
\end{proof}

\section{\revA{Tight} networks}

\label{sec:network-to-tight-network}

In this section we give the only lemma that is specific to three trees. We describe a transformation from a hybridization network to a structure, called a \revA{tight} network with embedded trees, that has some desirable properties. We prove that the transformation preserves the hybridization number, so we are allowed to concentrate on \revA{tight} networks in the case of three trees. \leo{For ease of notation, we will sometimes identify a directed graph with its edge set.}

A \emph{\revA{Tight} Network with Embedded Trees (\revA{TNET})} for a set $\T$ of phylogenetic trees over a label set $X$ is a pair $\H = (H, \E)$ with the following properties:
\begin{enumerate}[label=(\roman{*}),leftmargin=*,widest=viii,noitemsep]
\item $H$ is a DAG. We call its sources \emph{roots} and its sinks \emph{leaves}.\label{cond:acyclic-network}
\item Every root of $H$ has one child.\label{cond:non-splitting-roots}
\item The leaves of $H$ are labelled bijectively with the leaf labels in 
  $X$.\label{cond:leaf-labelling}
\item $\E = \set{H(T) \subseteq H \mid T \in \T}$ and, for all $T \in \T$,
  $H(T)$ is an \emph{image} of $T$, that is, $T$ can be obtained from $H(T)$
  by suppressing degree-$2$ nodes.\label{cond:displays-trees}
\item Every tree image $H(T) \in \E$ contains an edge incident to a root of
  $H$.\label{cond:start-at-the-root}
\item $H = \bigcup_{T \in \T}H(T)$, that is, every edge of $H$ belongs to at
  least one tree image in $\E$.\label{cond:no-useless-edges}
\item Every \leo{non-leaf} non-root node of $H$ has exactly two
  children.\label{cond:bifurcating-nodes}
\item For every \leo{non-leaf} non-root node, there exists a tree image
  $H(T) \in \E$ that contains both its child
  edges.\label{cond:every-node-splits}
\end{enumerate}

We represent the tree images in $\E$ by associating a unique colour with each tree $T \in \T$ and colouring every edge in $H(T)$ with this colour. We call the colour associated with tree $T$ \emph{colour $T$}. We use $C(e)$ to denote the colour set of an edge $e$ of $H$, that is, the set of trees $T \in \T$ whose images $H(T) \in \E$ include $e$. A \revA{TNET} for the input trees in Figure~\ref{fig:reconstruction-input} on page~\pageref{fig:reconstruction-input} is shown in Figure~\ref{fig:reconstruction-output} on page~\pageref{fig:reconstruction-output}. A corresponding hybridization network is shown in Figure~\ref{fig:reconstruction-final} on page~\pageref{fig:reconstruction-final}.

The hybridization network \emph{induced} by a \revA{TNET}~$(H', \E)$ is the hybridization network obtained by applying the following transformations to $H'$:
  \begin{itemize}[noitemsep]
  \item We replace every node $x$ that is both a reticulation node and a split
    node with two nodes $x_t$ and~$x_b$, change the bottom endpoints of $x$'s
    parent edges to $x_t$, change the top endpoints of $x$'s child edges to
    $x_b$, and add an edge from $x_t$ to $x_b$.
  \item We replace every reticulation node with more than two parents with a
    chain of binary reticulation nodes. \leo{See Figure~\ref{fig:nodesplit} for an illustration of these first two steps.}
  \item As long as there are at least two roots, we choose two such roots $r_1$
    and $r_2$, change the top endpoint of $r_1$'s child edge to $r_2$, and add
    an edge from $r_1$ to $r_2$. 
		This reduces the number of roots by one, so we eventually obtain a network
    with a single root. See Figure \ref{fig:rootjoin}.
\item For every leaf $x$ with more than one parent, we create a new
    node $x'$, change the bottom endpoint of every parent edge of $x$ to $x'$,
    and add an edge from $x'$ to $x$.
  \end{itemize}
	
\begin{figure}[t]
	\centering
		\includegraphics{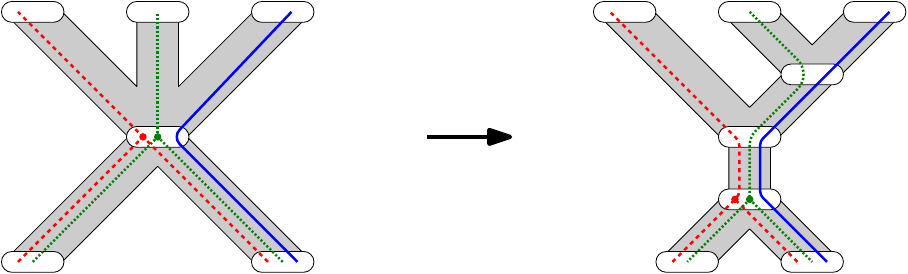}
	\caption{The first two steps of transforming a \revA{TNET} into a hybridization network: expanding nodes that are both reticulation and split nodes and refining reticulations.}
	\label{fig:nodesplit}
\end{figure}

\begin{figure}[t]
	\centering
		\includegraphics{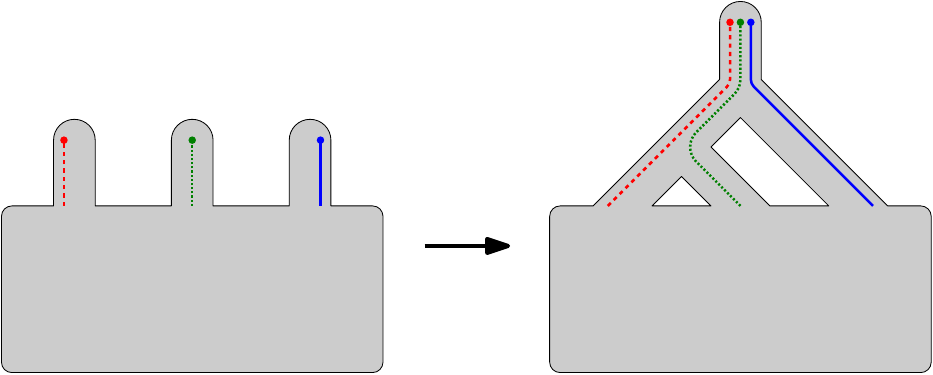}
	\caption{The third step of transforming a \revA{TNET} into a hybridization network: combining multiple roots.}
	\label{fig:rootjoin}
\end{figure}

The \emph{deletion AAF} of a \revA{TNET}~$(H, \E)$ is defined to be the deletion AAF of the hybridization network induced by~$(H, \E)$. \leo{The \emph{hybridization number} of a \revA{TNET}~$(H, \E)$ is (as for networks) defined to be the sum $\sum(d^-(v)-1)$ over all nodes~$v$ of~$H$ with indegree $d^-(v)$ at least~2.}

\begin{lem}
  \label{lem:multiroot-is-fine}
  If $\size{\T} = 3$, then there exists a hybridization network~$H$ with
  hybridization number $k$ for $\T$ if and only if there exists a
  \revA{TNET} $\H = (H', \E)$ with hybridization number $k$ for $\T$. Moreover, if~\revA{such a network~$H$ exists then there exists such a \revA{TNET} $\H$ with the same deletion AAF}.
\end{lem}

\begin{proof}
  First suppose that there exists a \revA{TNET} $\H = (H', \E)$ with hybridization number $k$ for $\T$. Then the hybridization network~$H$ induced by~$\H$ has the same hybridization number as $H'$, and it is easy to see that $H$ displays the trees in $\T$,   given that $H'$ displays these trees. It follows directly from the definition of the deletion AAF of a \revA{TNET} that~$H$ and~$\H$ have the same deletion AAF.

  Now assume we are given a hybridization network $H$ with hybridization number $k$ for $\T$. Since $H$ displays all trees in $\T$, we can choose a tree image $H(T)$, for every tree $T \in \T$, that includes the root of~$H$. Then we set $\H = (H, \E)$. $\H$ satisfies conditions~\ref{cond:acyclic-network}--\ref{cond:start-at-the-root} of a \revA{TNET} but may violate the remaining three conditions. Next we describe transformations that we apply to $\H$ to ensure it satisfies these remaining conditions without introducing any violations of the conditions $\H$ already satisfies and without increasing the hybridization number of $\H$. Thus, after applying these transformations, we obtain a \leo{\revA{TNET}} with hybridization number at most $k$ for~$\T$.

  \paragraph{Condition~\ref{cond:no-useless-edges}.}

  Deleting all edges of $H$ that are not contained in $\bigcup_{T \in \T}H(T)$
  does not violate
  conditions~\ref{cond:acyclic-network}--\ref{cond:start-at-the-root}\comment{LI}{Note that this cannot create new leaves because then the edge entering the new leaf would also have been deleted.} and
  establishes condition~\ref{cond:no-useless-edges}.
  Since it also does not increase the hybridization number of~$\H$, we obtain
  a network with hybridization number at most $k$ that satisfies conditions
  \ref{cond:acyclic-network}--\ref{cond:no-useless-edges}.

  \paragraph{Condition~\ref{cond:bifurcating-nodes}.}

  As long as there is a \emph{non-splitting} non-root node $x$, that is, a
  non-root node with only one child \revA{(which may arise after the modifications from the previous paragraph)}, we contract the edge $e$ between $x$ and
  its child in $H$ and in every tree image $H(T) \in \E$ that includes $e$, and
  merge any parallel edges this may create. Each such contraction reduces the number of non-splitting non-root nodes by
  one, does not introduce any violations of
  conditions~\ref{cond:acyclic-network}--\ref{cond:no-useless-edges}, and does
  not increase the hybridization number of~$\H$.
  Thus, we eventually obtain a network with hybridization number at most $k$
  that satisfies
  conditions~\ref{cond:acyclic-network}--\ref{cond:bifurcating-nodes}.

  \paragraph{Condition~\ref{cond:every-node-splits}.}

  By condition~\ref{cond:bifurcating-nodes}, every non-root node of $H$ is a
  split node.
  We call it a \emph{true split node} if it also satisfies
  condition~\ref{cond:every-node-splits}, and a \emph{fake split node}
  otherwise.
  We also call a true split node \leo{a \emph{$T$-split node} if the tree image~$H(T)$ contains both its child edges.}
  The \emph{weight} of a node $x$ is the number of trees $T \in \T$ such that
  $x$ is a $T$-split node. The \emph{weight} of a path is the sum of the weights of the nodes on the path. Now we define the \emph{potential} $\phi_x$ of a fake split node to be one plus the maximum weight of a path from a root to $x$. \leo{All} other nodes have potential $0$.
  The potential of the network is $\Phi := \sum_x \phi_x$, where the sum is
  taken over all nodes of the network.
  Since every fake split node has a positive potential, a network has potential $\Phi = 0$ if and only if it contains no fake split node, that is, if and only if it satisfies condition~\ref{cond:every-node-splits}.
  Next we describe transformations that decrease the potential of the network
  without increasing its hybridization number or introducing any violations
  of conditions~\ref{cond:acyclic-network}--\ref{cond:bifurcating-nodes}.
  Thus, \leo{after repeating this transformation as often as possible, we} obtain a network with
  hybridization number at most $k$ and which satisfies
  conditions~\ref{cond:acyclic-network}--\ref{cond:every-node-splits}, so it
  is a \revA{TNET} with hybridization number at most $k$
  for~$\T$.

  While the network contains fake split nodes, there exists such a node $x$ all
  of whose parents are true split nodes or roots.
  (Simply remove all roots, true split nodes, and leaves from $H$ and choose
  $x$ to be an in-degree-$0$ node of the resulting subgraph of $H$.)
  Since the colour sets of $x$'s child edges are disjoint and $x$ has two child
  edges, one of these edges, $e$, must have exactly one colour:
  $C(e) = \set{T}$, for some $T \in \T$.\footnote{This is the only place in the
    entire paper where we use that $\size{\T} = 3$.
    All other arguments are easily seen to generalize to more than $3$ trees.
    Alas, this argument is crucial because our algorithm does not work without
    condition~\ref{cond:every-node-splits}.}
  Let $f$ be $x$'s parent edge that has colour $T$.
  By the choice of $x$, the top endpoint $y$ of $f$ is a root or a true split
  node.

  \begin{figure}[t]
    \centering
    \includegraphics{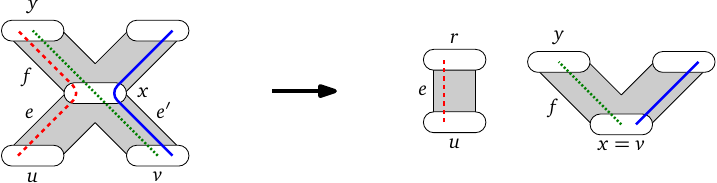}
    \caption{Eliminating a fake split node $x$ below a root $y$.}
    \label{fig:root-transformation}
  \end{figure}

  If $y$ is a root (see Figure~\ref{fig:root-transformation}), we remove $T$
  from the colour set of $f$, create a new root
  node $r$, change $e$'s top endpoint to $r$, remove $f$ and its top
  endpoint\leo{~$y$ if the colour set of~$f$} is now empty and restore condition~\ref{cond:bifurcating-nodes}.
  It is easily verified that this does not introduce any violations of
  conditions~\ref{cond:acyclic-network}--\ref{cond:no-useless-edges} and
  does not increase the hybridization number of the network.
  \leo{It also does not increase the potential of any node, and reduces the number of fake split nodes by one. The} potential of the network therefore decreases.

  \begin{figure}[t]
    \centering
    \subfigure[\unskip]{\includegraphics{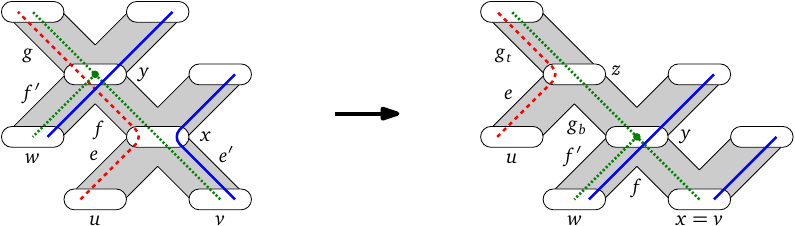}}
    \subfigure[\unskip]{\includegraphics{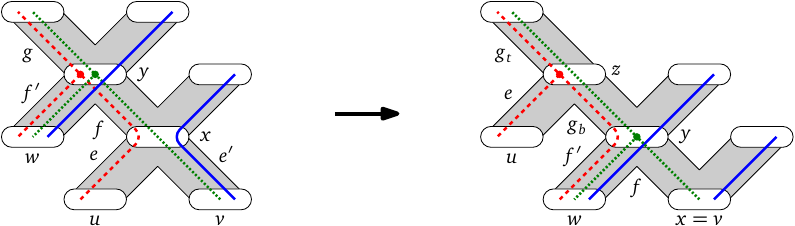}}
    \caption{Eliminating a fake split node $x$ below a split node $y$ \leo{that is a~$T'$-split node with~$T'$ not equal to} the colour~$T$ \leo{(red, dashed)} of the monochromatic edge below $x$.
      (a) $y$ is not a $T$-split node.
      (b) $y$ is a $T$-split node.}
    \label{fig:foreign-split-transformation}
  \end{figure}

  If $y$ is a true split node, we distinguish two cases.
  The first case is that~$y$ is a $T'$-split node, for some $T' \ne T$. This case is illustrated in Figure~\ref{fig:foreign-split-transformation}. Figure~\ref{fig:foreign-split-transformation}(a) depicts the subcase when~$y$ is not \leo{also} a~$T$-split node while Figure~\ref{fig:foreign-split-transformation}(b) depicts the subcase when~$y$ is also a~$T$-split node. Both cases can be handled in a similar way.

  Let $g$ be the
  parent edge of $y$ whose colour set includes $T$, and let $f'$ be $f$'s
  sibling edge.
  We divide $g$ into two edges $g_t$ and $g_b$, with $g_t$ above $g_b$, and
  denote their common endpoint by $z$.
  We remove $T$ from the colour set of $f$, set $C(g_t) := C(g)$ and
$C(g_b) := C(g) \setminus \set{T}$ if $T\notin C(f')$ and $C(g_b) := C(g)$ if $T\in C(f')$, change the top endpoint of $e$
  to $z$, remove all edges whose colour sets are now empty (this can only be $f$
  and $g_b$), and finally restore condition~\ref{cond:bifurcating-nodes}.
  Again, it is easily verified that this does not introduce any violations of
  conditions~\ref{cond:acyclic-network}--\ref{cond:no-useless-edges} and
  does not increase the hybridization number of the network.
  It also does not increase the potential of any node and eliminates $x$ from the network (because $x$ has only one child edge~$e'$ after changing $e$'s top endpoint and hence~$e'$ is being contracted).
  The only new node is $z$.
  If $T \in C(f')$, then $z$ is a true split node and its contribution to the network's potential is~$0$.
  Thus, since~$x$ is eliminated from the network, the network's potential decreases.
  If $T \notin C(f')$, then $z$ is a fake split node.
  However, its potential $\phi_z$ is less than $\phi_x$ \leo{because for every path from a root to~$z$ in the modified network, there exists a corresponding path from this root to~$x$ in the original network that has a greater weight because it contains the true split node~$y$.}
  Thus, once again, the potential of the network decreases.

  \begin{figure}[t]
    \centering
    \includegraphics{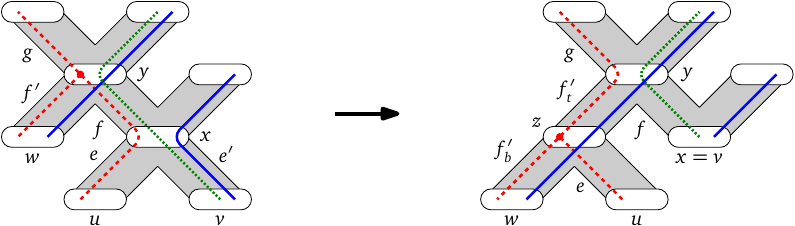}
    \caption{Eliminating a fake split node~$x$ below a split node $y$ \leo{that is not a~$T'$-split node for any~$T'$ that is not equal to the colour~$T$ (red, dashed) of the monochromatic edge below $x$.}}
    \label{fig:self-split-transformation}
  \end{figure}

  Finally, if $y$ is not a $T'$-split node for any $T' \ne T$ (see
  Figure~\ref{fig:self-split-transformation}), it must be a
  $T$-split node. As before, let $g$ be the
  parent edge of $y$ whose colour set includes $T$, and let $f'$ be $f$'s
  sibling edge. We subdivide~$f'$ into~$f'_t$ and~$f'_b$ where~$f'_t$ is above~$f'_b$ and let~$z$ be the newly created node. We change the top endpoint of~$e$ to~$z$, remove~$T$ from the colour set of~$f$, remove all edges whose colour sets are now empty, and restore condition~\ref{cond:bifurcating-nodes}. This transformation maintains conditions
  \ref{cond:acyclic-network}--\ref{cond:no-useless-edges} and does not
  increase the hybridization number of the network.
  It eliminates $x$ from the network (because it has only one child edge~\leo{$e'$} after
  changing the top endpoint of~\leo{$e$}) and makes~$y$ a fake split node.
  However, the potential of~$y$ in the modified network is less than the potential of~$x$ in the original network. To see this, let~$P$ be a path of maximum weight from a root to~$y$. Then the potential of~$y$ in the modified network is one plus the weight of~$P$. In the original network, the path~$P$ extended by the edge~$(y,x)$ is then a path from a root to~$x$, and its weight is one higher than the weight of~$P$ \leo{because it contains the true split node~$y$}. Hence, the potential of~$x$ in the original network is at least one higher than the potential of~$y$ in the modified network. Since the potential of all other nodes remains the same or decreases, the potential of the network decreases.
  
  Let $\H = (H', \E')$ be the \revA{TNET} eventually obtained by the above transformations and let~$H$ be the original hybridization network. It remains to show that~$\H$ and~$H$ have the same deletion AAF.
  \revA{To this end, observe that, if $F$ is a component of the deletion AAF of $H$ and $x, y$ are
    two taxa in $F$, then a network
    edge $e$ belongs to the path from $x$ to $y$ in the image $H(T)$ of some tree
    $T \in \T$ if and only if it belongs to the path from $x$ to $y$ in every image
    $H(T')$, $T' \in \T$.
    Now it suffices to verify that, in each of the transformations in Figures~\ref{fig:root-transformation}--\ref{fig:self-split-transformation}, none of the edges that
    are destroyed or created belongs to all three tree images, so each transformation leaves the deletion AAF unchanged.}
\end{proof}

By Lemma~\ref{lem:multiroot-is-fine}, it is sufficient if our algorithm
can construct a \revA{TNET} of the three input trees.
In the next section, we show how to do this.


\section{Reconstructing a \revA{tight} network}

\label{sec:guessing-things}

Let $\H = (H, \E)$ be a \revA{TNET} for $\T$, let
$F$ be its deletion AAF, and let $k$ be its hybridization number.
For each tree $T \in \T$, let $I(T)$ be the set of nodes in $T$
that do not belong to any path between two leaves~$x$ and~$y$ in the same AAF
component (considering the root as a leaf).
In a sense, these nodes are ``invisible'' in $F$.
The \emph{extended AAF} $F^*$ of $\T$ is defined as
$F \cup I$, where $I := \bigcup_{T \in \T} I(T)$. \leo{We will refer to the elements of~$F^*$ as \emph{components}. Let~$C$ be a component of~$F^*$. Hence,~$C$ is either an AAF component or a node in~$I$. If~$C$ is an AAF component, then~$r_C$ denotes the root of~$C$. If~$C$ is a node of~$I$, then~$r_C$ is equal to the node~$C$. In either case, we will refer to~$r_C$ as the \emph{root} of~$C$.}

By the following lemma, the size of~$I$ is at most~$3(k-1)$ if~$\T$ contains at most three trees.

\begin{lem}\label{lem:invisiblenodes}
For any~$T\in\T$, $|I(T)| \leq k-1$.
\end{lem}
\begin{proof}
Let~$T\in\T$. Since the root of~$T$ is a leaf after omitting directions, we can see~$T$ as an unrooted tree. Since~$F$ is an AAF of~$\T$ and~$T\in\T$, we know that~$F$ can be obtained by deleting a set~$E^*$ of~$k$ edges from~$T$ and then taking the partition of the leaves induced by the resulting connected components. Let~$\C$ be the partition of the nodes of~$I(T)$ induced by the connected components of the subgraph of~$T$ \revA{containing the vertices of~$I(T)$ and all edges between them that are not in~$E^*$. See Figure~\ref{fig:invisiblenodes} for an example}. For each~$C\in\C$, let~$\delta(C)$ denote the set of edges with exactly one endpoint in~$C$. Then, at most one of the edges in~$\delta(C)$ is not contained in~$E^*$ \revA{because otherwise at least one vertex of~$C$ would be on a path between two vertices in the same AAF component, which is not possible by the definition of~$I(T)$}. Hence, at least $|C| + 1$ of the~$|C|+2$ edges in~$\delta(C)$ are in~$E^*$. Moreover, since~$T$ is a tree, at most~$|\C|-1$ edges can be in $\delta(C)\cap\delta(C')$ for two different~$C,C'\in\C$. Hence,
\[
|E^*| \geq \sum_{C\in\C} (|C|+1) - (|\C|-1) = |I(T)| + 1.
\]
Since~$|E^*|=k$, the lemma follows.
\end{proof}

\begin{figure}[t]
  \centering
  \includegraphics{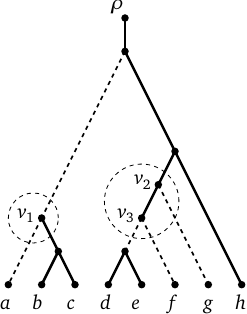}
  \caption{\revA{Example for the proof of Lemma~4. The set~$E^*$ consists of the $k=5$ dashed edges. The two elements of~$\C$ are indicated by dashed circles.}}
  \label{fig:invisiblenodes}
\end{figure}

We \leo{will} construct $\H$ from $F^*$ and $\T$ with the help of a
\emph{guess} of the structure of $\H$.
In particular we construct $\H$ by ``gluing together'' the components
of $F^*$. Our guess concerns how this gluing is to be done.
\revA{Under the embedding of~$\T$ in~$H$ described by~$\E$,} the root $r_C$ of every component $C \in F^*$ has a unique image $H(r_C)$ in
$H$:
If $C \in F$, this is true because $F$ is the deletion AAF of $\H$.
If $r_C = C \in I(T)$, for some $T \in \T$, $r_C$ is a split node of
\leo{$T$ and, thus, has a unique image $H(r_C)$ in $H(T)$ which is also a node of~$H$}.
Our guess for $r_C$ defines the ``wiring'' of the in-edges of $H(r_C)$;
we call it the \emph{wiring guess} for~$r_C$.
Since the colour sets of the parent edges of every node in $H$ are disjoint,
$H(r_C)$ has between one and three parent edges.
The first part of the wiring guess for $r_C$ is the number of parent edges of
$H(r_C)$.
The second part of the wiring guess for $r_C$ is which of these in-edges is
included in which tree image.
Finally, observe that the top endpoint~$x$ of each parent edge of~$H(r_C)$ must
once again be a $T'$-split node, for at least one $T' \in \T$.
The third part of the wiring guess for $r_C$ determines such a tree $T'$ for
each parent edge of $H(r_C)$. \leo{We will assume without loss of generality that any two nodes~$r_C,r_{C'}$ for which $H(r_C)$ and $H(r_{C'})$ have a common parent~$x$ both guess the same tree~$T'$ as the tree for which~$x$ is a~$T'$-split node.}

\leo{First consider a component root~$r_C$ that is a node in~$I(T)$ for some~$T\in\T$.}
Note that (i) at least one parent edge of $H(r_C)$ must have colour $T$ because
$H(r_C)$ is \leo{a non-root node} of $H(T)$, and (ii) the top endpoint
of a parent edge $e$ of $H(r_C)$ can be a $T'$-split node only for trees $T'$
such that $T' \in C(e)$.
This gives the~$17$ possible wiring guesses for~$r_C$ shown in Figure~\ref{fig:guesses}.

If $r_C$ is an AAF root, the set of possible wiring guesses is more restricted.
Since $H(r_C)$ is contained in every tree image, the only valid wiring guesses
are the ones where the union of the colour sets of the parent edges of $H(r_C)$
contains all three colours. This reduces the number of possible wiring guesses for AAF roots to $10$ \leo{(see again Figure~\ref{fig:guesses})}. \leo{Finally note that, when~$r_C$ is the root~$\rho$ of the trees, there is only a single wiring guess.}

Our \emph{guess} $\G$ of the structure of $\H$ consists of the wiring guesses for all roots $r_C$, $C \in F^*$.

\begin{figure}[t]
  \centering
  \includegraphics{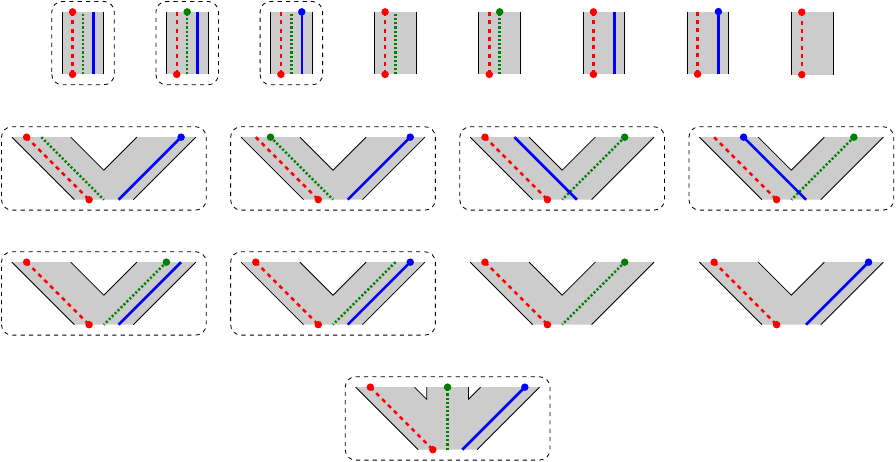}
  \caption{The 17 possible wiring guesses for \leo{the root~$r_C$ of a component~$C$ of the extended AAF~$F^*$ that is a node in~$I(T)$, with~$T$ the red dashed tree}.
    The guess of the split node at the top of each edge is indicated by a dot.
    Valid guesses for a root~$r_C$ of an AAF component are indicated by dashed boxes around them.}
  \label{fig:guesses}
\end{figure}

Since we have $17$ wiring guesses to choose from for each component in $I$,
and $10$ wiring guesses to choose from for each component in $F$ that does not
contain the tree root $\rho$, there are $10^{\size{F}-1} \cdot 17^{\size{I}}$
possible guesses $\G$.
Since $\size{F} \le k+1$ and $\size{I} \le 3(k-1)$ by Lemma~\ref{lem:invisiblenodes}, the number of possible guesses $\G$ is bounded
by $10^k \cdot 17^{3(k-1)} = 49130^k / 4913$.

Our algorithm considers each guess \leo{in}~$\G$ in turn and attempts to construct a
\revA{TNET} $\H$ from $(\G, F^*, \T)$ in polynomial time.
We call $(\G, F^*, \T)$ the \revA{TNET}'s \emph{description}.
To establish our algorithm's correctness, we prove that, given the description
$(\G, F^*, \T)$ of a \revA{TNET} $\H$ for $\T$, our
algorithm succeeds in constructing a \revA{TNET} $\H'$ for
$\T$ with description $(\G, F^*, \T)$ that differs from $\H$ only in
insignificant details and in particular has the same hybridization number
as~$\H$. Our proof is divided into two lemmas. The first one shows that two \revA{TNET}s with the same description are essentially
the same.
We make this precise below.
The second one shows that, given the description of a \revA{TNET}, we can construct a \revA{TNET} with this description in polynomial time.
Note that not every description $(\G, F^*, \T)$ is necessarily a valid
description of any \revA{TNET}.
\leo{If there is no \revA{TNET} with the given description, then our algorithm will determine this in polynomial time.} The description of this algorithm is given in the proofs of
Lemmas~\ref{lem:unique-network} and~\ref{lem:network-from-description} below.
Appendix~\ref{sec:example} provides an example of the operation of the
algorithm.

Given a \revA{TNET} $\H = (H, \E)$ for $\T$, we obtain a
DAG $\tilde H$ by contracting the image of every component of the
deletion AAF $F$ of $H$ into a single node, keeping parallel edges this
creates.
We call $\tilde H$ the \emph{signature} of $H$. \leo{For an example, see Figure\revA{s}~\ref{fig:reconstruction-signature} \revA{and~\ref{fig:reconstruction-output}} in the appendix.}
It is easy to see that the node set of $\tilde H$ is
$\set{H(r_C) \mid C \in F^*}$ and that $\tilde H$ has the same hybridization
number as $H$. Note that two nodes~$C\in I(T)$, $C'\in I(T')$ might map to the same node~$H(r_C)=H(r{_C'})$ of~$\tilde H$.
We label every node $x$ of $\tilde H$ with the set of roots $r_C$, $C \in F^*$,
such that $x = H(r_C)$.
We call roots that label the same node of~$\tilde H$ \emph{buddies} (of each other). Note that roots of components of~$F$ have no buddies.


We will use the following notion of the \leo{``attached subtrees''} of an AAF component~$C$ in a tree~$T$. For each edge~$(u,v)$ of~$T$ for which~$u$ is contained in~$T(C)$ but~$v$ is not contained in~$T(C)$, we say that the subtree of~$T$ induced by~$u,v$ and all nodes reachable from~$v$ is \leo{a subtree of~$T$ \emph{attached} to~$C$.}
  
In addition, we will use the following notion of the ``pendant subtrees'' represented by an edge~\leo{$e=(H(r_C),H(r_{C'}))$ of the signature~$\tilde H$. For each~$T\in C(e)$, the subtree of~$T$ induced by~$r_{C'}$, the parent of~$r_{C'}$ and all nodes reachable from~$r_{C'}$ is} said to be a \emph{pendant subtree} of~$T$ represented by the edge~$e$ of~$\tilde H$.

\begin{lem}
  \label{lem:unique-network} \revA{The signature $\tilde H$ of a \revA{TNET} $\H = (H, \E)$ can be uniquely reconstructed in polynomial time from the description $(\G, F^*, \T)$ of $\H$.}
\end{lem}

\begin{proof}
  We define a DAG $D$ whose nodes are the roots of \leo{components of}~$F^*$.
  There is an edge from a root $r_C$ to a root $r_{C'}$ if and only if there
  exists a tree $T \in \T$ such that $r_{C'}$ is an ancestor of $r_C$ in $T$
  and there exists no component root $r_{C''}$ that is an ancestor of $r_C$ and
  a descendant of $r_{C'}$ in~$T$.
  If this is the case, we call $r_C$ a \emph{direct descendant} of $r_{C'}$ in $T$. (Note that the edge in~$D$ is directed from the descendant to the ancestor.)
  We incrementally construct a topological ordering of $D$, define $U_i$ to be the first~$i$ nodes in this topological ordering, $D_i$ to
  be the subgraph of $D$ induced by $U_i$, and $V_i := \{ H(r_C) \mid r_C \in U_i\}$.
  $\bar D_i$~is the subgraph of $D$ induced by all nodes of $D$ not in $U_i$. Let \leo{the partial signature} $\tilde H_i$ be the graph obtained as follows from the subgraph of $\tilde H$ induced by~$V_i$:
For each parent edge~$e$ in~$\tilde H$ of a node $H(r_C) \in V_i$ whose top endpoint does not belong to $V_i$, we adding a new root and an edge~$\tilde e$ from this new root to~$H(r_C)$ and give~$\tilde e$ the same colour set as~$e$ (which is determined by the wiring guess for $r_C$). \leo{For an example, see Figure~\ref{fig:reconstruction-signature} in the appendix.} We call an edge whose top endpoint is a root a \emph{root edge}.
	
  

  We assign these new roots colours, where the colour of the top endpoint
  of edge $\tilde e$ is $T$ if the bottom endpoint of $\tilde e$ guesses that the top endpoint of $e$ is a $T$-split node.
  
We will show that, from~$D_i, \tilde H_i$ and~$\G$, we can determine a unique nonempty set of nodes~$U'$ of $\bar D_i$, the set~$U^+$ of all buddies of nodes in~$U'$ and \leo{an extended partial signature}~$\tilde H_j$ corresponding to the graph~$D_j$ induced by~$U_j:=U_i\cup U' \cup U^+$.
  
  
  
  Once $D_i = D$, we thus obtain $\tilde H_i = \tilde H$, that is, we are able
  to reconstruct \leo{the signature} $\tilde H$ only given the description of~$\H$.
  
We initialize the reconstruction by defining $D_0$ and~$\tilde H_0$ to be empty \leo{digraphs}. Now consider an iteration with input \leo{digraph} $D_i$ and \leo{partial signature} $\tilde H_i$. 
  For an in-degree-$0$ node $r_C$ of $\bar D_i$ that belongs to $I(T)$, for some
  $T \in \T$, observe that both child edges of $H(r_C)$ in $\tilde H$ have corresponding root edges in $\tilde H_i$ (by the construction of $\tilde H_i$). For such a node~$r_C$, we call $r_C$ \emph{free} if the top endpoints of both of these root edges are coloured $T$.
  For an in-degree-$0$ node $r_C$ of $\bar D_i$ that is a root of \leo{a component of} $F$, we say that~$r_C$ is \emph{free}
  if, for every in-edge $e$ of $r_C$ in $D$, the corresponding root edge $\tilde e$ of $\tilde H_i$ has the property that, for every
  $T \in C(\tilde e)$, the \leo{pendant} subtree of $T$ represented by $\tilde e$ is \leo{attached to the AAF component $C$ in $T$}. All other nodes of $\bar D_i$ are non-free.
  We claim that at least one of the nodes in $\bar D_i$ is free and that every
  free node can determine its buddies in $\tilde H$ and can augment \leo{the partial signature} $\tilde H_i$
  to $\tilde H_j$.

  Consider a node $x$ of \leo{the signature} $\tilde H$ that does not belong to $\tilde H_i$ and all of whose children \leo{do} belong to $\tilde H_i$.
  Since $\tilde H$ is a DAG, such a node exists.
  Node $x$ is a $T$-split node, for some $T \in \T$, and is thus the image
  $H(r_C)$ of \leo{the root~$r_C$ of a component} in~$F \cup I(T)$.
  
  First consider the case that \leo{$C \in F$}. Observe that, because $\tilde H$ is the signature of a \revA{TNET} $\H$ for $\T$, the \leo{subtrees attached to} $C$ in all input trees are exactly the \leo{pendant} subtrees represented by the child
edges of $H(r_C)$ in $\tilde H$. Since these child edges are root edges of $\tilde H_i$, $r_C$ is free. Node $r_C$ can trivially identify its set of buddies because it has no buddies besides itself. We obtain~$D_j$ by adding $r_C$ to $D_i$ and obtain~$\tilde H_j$ from~$\tilde H_i$ by locating the root edges of $\tilde H_i$ that correspond to child edges of $H(r_C)$ in $\tilde H$, merging the top endpoints of these root edges to a single node~$H(r_C)$, and adding new root edges entering~$H(r_C)$ according to~$r_C$'s wiring guess.
  This gives a new graph $D_j \supset D_i$ and the corresponding graph
  $\tilde H_j$.

  Now consider the case that $r_C \in I(T)$. In this case, both child edges of~$r_C$ in $\tilde H$ correspond to root edges of $\tilde H_i$.
  Since these two edges have a common top endpoint in $\tilde H$, their
  top endpoints in $\tilde H_i$ must have the same colour $T'$.
  This implies that $H(r_C)$ is a $T'$-split node and is thus the image
  $H(r_{C'})$ of a node $r_{C'} \in I(T')$.
  Since the child edges of $H(r_{C'})$ in~$\tilde H$ belong to $\tilde H_i$, both
  in-neighbours of $r_{C'}$ in $D$ belong to~$D_i$.
  Thus, $r_{C'}$ is free.
  Since $H(r_C) = H(r_{C'})$, we may assume that $r_{C'} = r_C$ (because in at least one set of guesses this is the case). We obtain $\tilde H_j$ from~$\tilde H_i$ by locating the two root edges of~$\tilde H_i$ that correspond to child edges of $H(r_C)$ in $\tilde H$, merging the top endpoints of these edges to a single node~$H(r_C)$, and adding new root edges entering~$H(r_C)$ according to~$r_C$'s wiring guess.
  Now consider the child edges of $H(r_C)$ in $\tilde H_j$.
  For each tree $T' \in \T$ such that these edges are both coloured $T'$,
  $H(r_C)$ is the image $H(r_{C'})$ of a node in $I(T')$.
  The nodes of $\bar D_i$ that satisfy this condition are the buddies of
  $r_C$, and we add them to $D_i$ along with $r_C$ to obtain $D_j$.

  We have shown that every node $H(r_C)$ of $\tilde H$ whose children belong to
  $\tilde H_i$ corresponds to a free node $r_C$ in $\bar D_i$ and that $r_C$
  can determine its set of buddies and can construct the \leo{extended partial signature} $\tilde H_j$
  corresponding to the \leo{digraph} $D_j \supset D_i$ obtained by adding these buddies to $D_i$.
  If $r_C$ is an AAF root, it has no buddies besides itself.
  If $r_C \in I(T)$, for some $T \in \T$, observe that none of its buddies is
  free because it belongs to a tree $T' \ne T$ but the top endpoints of the
  child edges of $H(r_C)$ have colour $T$ in $\tilde H_i$.
  Thus, to prove that \emph{every} free node can determine its set of buddies
  and can construct $\tilde H_j$ from $\tilde H_i$, it suffices to show
  that there is no free node $r_C$ such that $H(r_C)$ has a child not in
  $\tilde H_i$.

  Assume there exists such a node $r_C$.
  \leo{Then} $r_C$ has in-degree $0$ in $\bar D_i$ because otherwise it
  is not free.
  First assume $r_C$ \leo{is the root of a component} in~$F$.
  Let $H(r_{C'})$ be an in-neighbour of $H(r_C)$ in $\tilde H$ that does not
  belong to~$\tilde H_i$, and let $T$ be an arbitrary colour $T$ in the
  colour set of the edge $e$ between $H(r_C)$ and $H(r_{C'})$ in~$\tilde H$.
  Since $r_C$ has in-degree~$0$ \leo{in $\bar D_i$}, the \leo{pendant} subtree of $T$ represented by $e$ is also represented by some root edge $f$ of $\tilde H_i$.
  Thus, $\tilde H$ contains a unique path from $e$ to $f$ and every node on this
  path is a $T'$-split node, for some $T' \ne T$.
  This implies in particular that the top endpoint of edge $f$ is a $T'$-split
  node, \leo{$T,T'\in C(f)$}, and the \leo{pendant} subtree of $T'$ represented by $f$ is not a \leo{subtree attached to $C$ in~$T'$}. Thus, since $f$ represents the same in-edge of $r_C$ in $D$ as $e$, $r_C$
  is not free, a contradiction.

  If $r_C \in I(T)$, for some tree $T \in \T$, we choose an in-neighbour
  $H(r_{C'})$ and a root edge $f$ of $\tilde H_i$ as in the case when
  $r_c \in F$.
  Since the top endpoint of edge $f$ is not a $T$-split node, its colour in
  $\tilde H_i$ must be $T' \ne T$.
  Thus, since $f$ is the root edge of $\tilde H_i$ representing one of the
  in-edges of $r_C$ in $D$, $r_C$ is not free, again a contradiction.
\end{proof}

By Lemma~\ref{lem:unique-network}, it suffices to provide a polynomial-time
algorithm that decides whether there exists a \revA{TNET}
with description $(\G, F^*, \T)$ and, if so, construct any such \revA{TNET}.
Our next lemma states that such an algorithm exists.

\begin{lem}
  \label{lem:network-from-description}
  Given a description $(\G, F^*, \T)$, it takes polynomial time to decide
  whether there exists a \revA{TNET} with this description
  and, if so, \leo{to} construct such a \revA{TNET}.
\end{lem}

\begin{proof}
  \revA{By Lemma~\ref{lem:unique-network}, there exists a polynomial-time algorithm
  for constructing the signature $\tilde H$ of a \revA{TNET} $\H = (H, \E)$ with
  description $(\G, F^*, \T)$, if such a \revA{TNET} exists.}
  This algorithm provides a first test that can be used to reject invalid
  descriptions:
  If the algorithm reaches an iteration where $\bar D_i$ is non-empty but none
  of its nodes is free, then the description is invalid because
  the proof of Lemma~\ref{lem:unique-network} shows that, if $(\G, F^*, \T)$
  is the description of a \revA{TNET}, then there exists
  a free node in each iteration.
  Thus, the algorithm aborts and rejects the description.
  If the algorithm does not reject the description, its output is a signature
  $\tilde H$ that respects all wiring guesses in $\G$.
  However, this signature may not correspond to a network $H$ that displays
  all input trees.
  Next we present a polynomial-time algorithm for constructing $\H = (H, \E)$
  from $(\tilde H, F^*, \T)$ or determining that no \revA{TNET} $\H = (H, \E)$
  with description $(\G, F^*, \T)$ exists.

  The nodes of $\tilde H$ are of two types: images of AAF roots and images of
  \leo{nodes in~$I$.}
  To obtain a \revA{TNET} $\H = (H, \E)$ with description $(\G, F^*, \T)$ from
  $\tilde H$, \leo{we let~$H$ initally be equal to~$\tilde H$ and} replace each AAF root image with the AAF component it represents.
  We process these AAF root images bottom-up, that is, in reverse topological
  order.\footnote{This is not really essential, but it simplifies the
    description of the algorithm.}
  Consider such an image $H(r_C)$ of the root $r_C$ of an AAF component $C$, and
  let $E$ be the set of child edges of $H(r_C)$ in $\tilde H$.
  We remove the edges in $E$ from $\tilde H$ and attach $C$ below $r_C$, setting
  $C(e) = \T$ for every edge $e$ of $C$.
  Our goal now is to reattach the edges in $E$ to edges of $C$ so that, for all
  $T \in \T$, the descendant edges of $H(r_C)$ with colour $T$ form an image
  of the \leo{pendant subtree of~$T$ with root~$r_C$}.
  For each tree $T$, observe that the edges in $E$ coloured $T$ represent the
 \leo{subtrees attached to} $C$ in $T$.
  We need to attach these edges to $C$ \leo{in~$H$} so that each edge branches off the same edge of $C$ as in $T$ and edges that branch off the same edge of $C$ in $T$ do so in the same order as in $T$.

  First we test every edge $e \in E$ whether $e$ branches off the same edge of
  $C$ in every tree \leo{$T\in C(e)$}.
  If this is the case for all edges $e \in E$, then we can partition $E$ into
  subsets $E_f$, one per edge $f$ of $C$, such that all edges in $E_f$ branch
  off edge $f$.
  If there is an edge $e$ that branches off some edge $f$ in a tree $T \in C(e)$
  and off a different edge $f'$ in a tree $T' \in C(e)$, it is impossible to
  attach this edge to $C$ in a way that satisfies both constraints.
  Since the edges of $T$ and $T'$ represented by $e$ are determined by
  $\tilde H$,
  which in turn is uniquely defined by $(\G, F^*, \T)$, there is thus no network
  $\H$ with description $(\G, F^*, \T)$, so the algorithm aborts and reports
  that there is no such network.

  Given the partition of $E$ into subsets $E_f$ such that the edges in $E_f$
  branch off edge $f$, we need to attach the edges in each such set $E_f$
  in an ordering consistent with the input trees.
  Let $E_{f,T}$ be the subset of edges in $E_f$ that have colour $T$.
  Tree $T$ determines the ordering in which these edges need to be attached to
  $f$. We define a DAG $D_f$ whose nodes represent the edges in $E_f$, and \leo{which has an edge~$(g,g')$ precisely if there is a tree~$T$ such that~$g,g'\in E_{f,T}$ and~$g$ branches off~$f$ in~$T$ above~$g'$.}
  There exists an ordering in which to attach \leo{the edges of~$E_f$} to $f$ so that
  the ordering constraints imposed by all trees in $\T$ are satisfied if and
  only if $D_f$ is acyclic. \leo{Moreover, if~$D_f$ is indeed acyclic, then a topological ordering of $D_f$ provides a valid ordering in which the edges can be attached.}
  Thus, we can test whether such an ordering exists and, if so, compute such
  an ordering in time $\OhOf{\size{E_f}}$ per edge $f$.
  If no such ordering exists, the algorithm once again aborts and reports that
  there is no \revA{TNET} with description $(\G, F^*, \T)$.
  If we can find a correct ordering of the edges attached to each edge $f$
  of $C$, then the replacement of $r_C$ with $C$ in this manner results in
  a network where all descendant edges of $H(r_C)$ with colour $T$ form
  an image of \leo{the pendant subtree of~$T$ with root~$r_C$}, for all $T \in \T$.
  Thus, after replacing each AAF root $r_C$ with its corresponding component
  $C$ in this fashion, we obtain a \revA{TNET} $\H = (H, \E)$ with
  description $(\G, F^*, \T)$.
\end{proof}

\revA{To summarize, the overall strategy for trying to reconstruct a hybridization network~$H$ with hybridization number~$k$ for an input set~$\T$ of three trees is as follows. First, we guess the deletion AAF~$F$ of~$H$ in time $\OhOf{ 2^{5k} \cdot 2^{10k} \cdot \poly{n} }$. Then we identify the set~$I$ of invisible nodes and add them to~$F$, obtaining the extended AAF~$F^*$, and guess $\G$, the wiring of each component of~$F^*$. The total number of possible guesses for~$\G$ is $10^k \cdot 17^{3(k-1)} = 49130^k / 4913$. For each possible description~$(\G, F^*, \T)$, we try to construct the signature $\tilde H$ of a \revA{TNET} with this description using Lemma~\ref{lem:unique-network} in polynomial time. For each signature $\tilde H$, we decide whether there exists a \revA{TNET} with this description (see Lemma~\ref{lem:network-from-description}), again in polynomial time. Once a correct TNET has been found, it can be expanded to the hybridization network~$H$ (see Section~\ref{sec:network-to-tight-network}). The overall running time is $\OhOf{ 2^{5k} \cdot 2^{10k} \cdot 10^k \cdot 17^{3(k-1)} \cdot \poly{n} }$, which is $\OhOf{  1,609,891,840^k \cdot \poly{n} }$.
}


\section{Conclusions}\label{sec:concl}

\revA{For two trees, a hybridization network can easily be constructed in polynomial time once the AAF is known. No guessing is required since the AAF carries all the necessary information. For more than two trees,
it seemed natural enough to try to guess the wiring structure that determines how the AAF components need to be glued together into a network. For any constant number of trees, there are only a constant number of choices for the wiring of the root of each component, so with $\OhOf{k}$ components, one would obtain an $\OhOf{c^k \cdot \poly{n}}$-time algorithm.}
Unfortunately, guessing the wiring structure of AAF components turned out
not to be enough, even for three trees, because there are examples of three
input trees such that every optimal network displaying these trees contains
an \emph{invisible component}: a group of nodes that are isolated from all
taxa once all hybridization edges are deleted, see Figure~\ref{fig:invisible} \leo{in the introduction}.
We call these components invisible because they are not represented in any
form in the AAF.

Guessing the number and structure of these invisible components seems extremely
challenging.
In this paper, we showed that one can get away without having to guess these
components in the case of three trees because, for three trees, these components
are not invisible after all: They may not be represented in the AAF, but they
are present as nodes in the three input trees, at least as long as we consider
only tight networks (and we have shown that we may do this without loss of generality).
This is the key to our $\OhOf{c^k \cdot \poly{n}}$-time algorithm for three
trees.
Unfortunately, it appears that we simply scraped by.\comment{LI}{Scraped by???}
While the framework of our algorithm extends to more than three trees, it seems
that already for four trees, there are input instances where the optimal
network includes truly invisible nodes: nodes whose only purpose is to change
the way in which edges of the tree images are braided together along network
edges, see Figure~\ref{fig:braided}.
Thus, the main open problem is to discover structural properties that, while
unlikely to eliminate the need to guess the existence of these braiding
structures in the network altogether, at least limit the number of possible
guesses to be explored.

  \begin{figure}[t]
    \centering
    \includegraphics{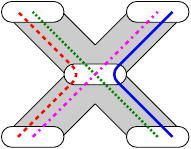}
    \caption{A ``truly invisible'' node of a network with four trees embedded in it. None of the four trees branches in this network node.}
    \label{fig:braided}
  \end{figure}

Another interesting question that arises from our work is whether guessing
the wiring of the extended AAF components as part of the reconstruction is
necessary at all, at least in the case of three trees.
Since all these components are visible in the input trees, it could be possible
that one can construct the entire network directly from the AAF, that is,
as in the case of two trees, the hard core of the problem is finding the right
AAF.
However, we conjecture that this is not the case:
i.e., that even given the deletion AAF of an optimal hybridization network
for the three input trees, it remains NP-hard to find the network.

\bibliographystyle{plain}
\bibliography{bibliographyleo}

\appendix

\section{An Example of constructing a network from its description}

\label{sec:example}

This appendix provides an example of the construction described in the proofs of
Lemmas~\ref{lem:unique-network} and~\ref{lem:network-from-description}.
Consider the description $(\G, F^*, \T)$ in
Figure~\ref{fig:reconstruction-input}.
The \revA{TNET} in Figure~\ref{fig:reconstruction-output} has this description.
We first show how to construct its signature using the construction in the
proof of Lemma~\ref{lem:unique-network}.
This signature is shown in Figure~\ref{fig:reconstruction-signature}.
Then we discuss how to construct the \revA{TNET} in
Figure~\ref{fig:reconstruction-output} from this signature.
\leo{The hybridization network induced by this \revA{TNET}} is shown in Figure~\ref{fig:reconstruction-final}.

\begin{figure}[b]
  \centering
  \includegraphics{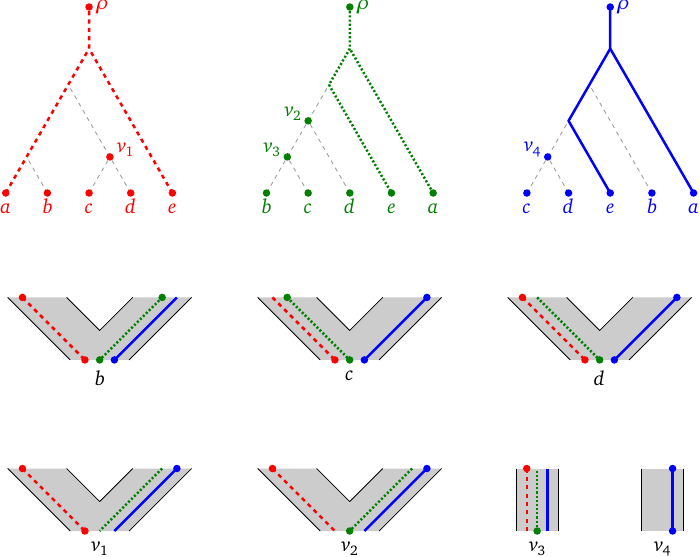}
  \caption{The input to the network reconstruction, including the AAF $F$ (shown
    in bold in the three input trees), \leo{the set of ``invisible nodes'' $I=\{v_1,\ldots ,v_4\}$}, and the wiring guesses
    for the resulting set of components of $F^*$. There is no guess for the component $\set{a, e, \rho}$ because it includes
    the root~\leo{$\rho$} of the three trees.}
  \label{fig:reconstruction-input}
\end{figure}

\paragraph{Lemma~\ref{lem:unique-network}: Constructing the signature.}

The DAG $D$ used in the construction of the signature $\tilde H$ of any \revA{TNET}
with the description in Figure~\ref{fig:reconstruction-input} is shown in
Figure~\ref{fig:signature-dag}.
This DAG represents the adjacency of components of $F^*$ in the three input
trees in $\T$.

\begin{figure}[t]
  \centering
  \includegraphics{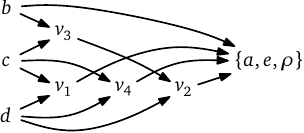}
  \caption{The DAG $D$ used in the reconstruction of the signature $\tilde H$
    in Figure~\ref{fig:reconstruction-signature} from the description in
    Figure~\ref{fig:reconstruction-input}.}
  \label{fig:signature-dag}
\end{figure}

\begin{figure}[t]
  \centering
  \includegraphics{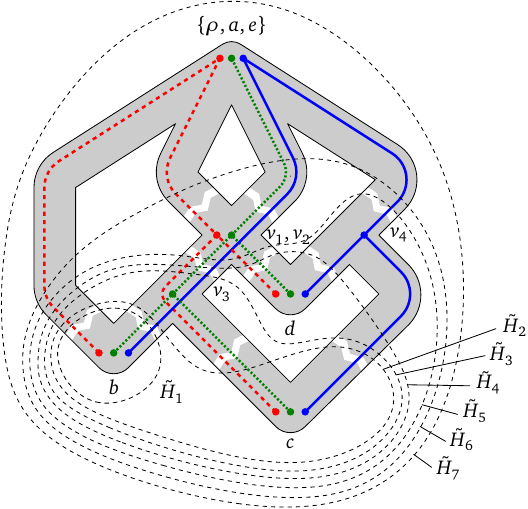}
  \caption{The signature $\tilde H$ of any \revA{TNET} with the description
    $(\G, F^*, \T)$ shown in Figure~\ref{fig:reconstruction-input}.
    The \leo{partial signatures} $\tilde H_1, \tilde H_2, \dots, \tilde H_7 = \tilde H$
    constructed incrementally are indicated by dashed lines \revA{(edges that are partly in the indicated region are also contained in the partial signature)}.}
  \label{fig:reconstruction-signature}
\end{figure}

\begin{figure}[t]
  \centering
  \includegraphics{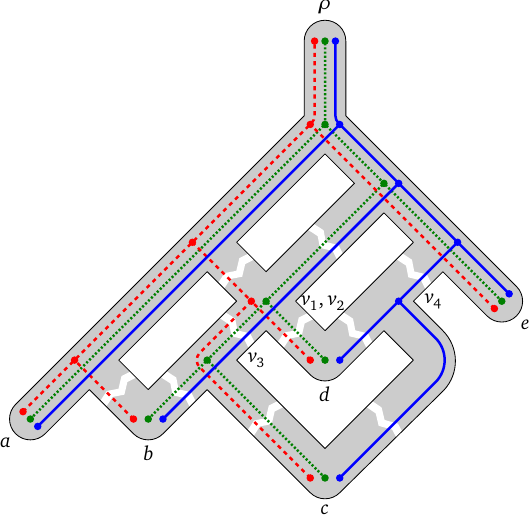}
  \caption{The \revA{TNET} obtained from the signature in
    Figure~\ref{fig:reconstruction-signature} by expanding the components
    of~$F$.}
  \label{fig:reconstruction-output}
\end{figure}

\begin{figure}[t]
  \centering
  \includegraphics{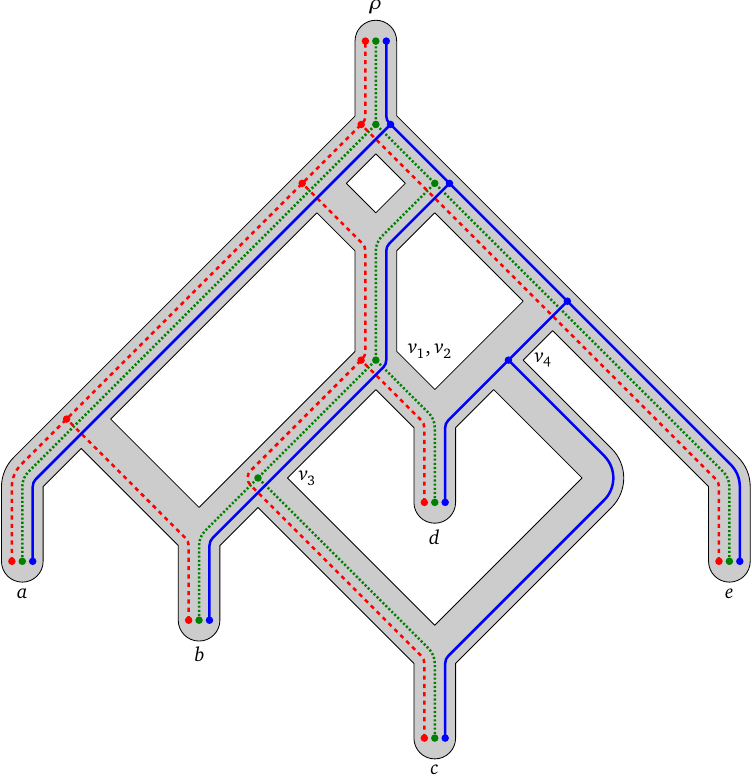}
  \caption{The hybridization network induced by the \revA{TNET} in
    Figure~\ref{fig:reconstruction-output}, obtained by separating reticulations,
split nodes, and leaves.} \label{fig:reconstruction-final}
\end{figure}

Before the first iteration, we have $\bar D_0 = D$.
The only nodes with in-degree $0$ in $\bar D_0$ are $b$, $c$, and $d$.
They are all free because they have no in-edges at all.
Assume we pick $b$ as the first node to add to $\tilde H$.
We create the node $H(b)$ and add parent edges according to $b$'s wiring guess
to obtain the partial network $\tilde H_1$ in
Figure~\ref{fig:reconstruction-signature}.
Since $b$ is the root of an AAF component, it has no buddies, so
$D_1$ has the node set $\set{b}$.

$\bar D_1$ has two nodes of in-degree $0$, namely $c$ and $d$.
Both are again free.
Assume we choose $c$ as the next node to add to $D_1$ to obtain $D_2$.
Since $c$ is again the root of an AAF component, it has no buddies, so the
node set of $D_2$ is $\set{b, c}$.
We construct the graph $\tilde H_2$ in Figure~\ref{fig:reconstruction-output}
from $\tilde H_1$ by creating a node $H(c)$ and \leo{adding} parent edges of this node
according to $c$'s wiring guess.

$\bar D_2$ has two nodes of in-degree $0$, namely $d$ and $v_3$.
Both nodes are free: $d$ is free \leo{because it is the root of an AAF component}; $v_3$ is free because $\tilde H_2$ has two root
edges above $H(b)$ and $H(c)$, which are children of $v_3$ in the green dotted
tree, and the top endpoints of both edges are coloured green by the wiring guesses
for $b$ and $c$.
Let us assume we choose $v_3$ as the next node to add to $D_2$ to obtain
$D_3$.
We create the node $H(v_3)$ by identifying the top endpoints of the two green
dotted parent edges of $H(b)$ and $H(c)$ and add a parent edge above $H(v_3)$
according to $v_3$'s wiring guess.
This gives the graph $\tilde H_3$ shown in
Figure~\ref{fig:reconstruction-signature}.
Since the only tree common to the colour sets of the parent edges of $H(b)$
and $H(c)$ is the green dotted one, $v_3$ has no buddies.
Thus, $D_3$ has the node set $\set{b, c, v_3}$.

$\bar D_3$ has $d$ as its only in-degree-$0$ node, and $d$ is free.
We obtain $D_4$ by adding $d$ to $D_3$.
Node $d$ has no buddies because it is the root of an AAF component.
To obtain $\tilde H_4$ from $\tilde H_3$, we create the node
$H(d)$ and add parent edges according to $d$'s wiring guess.

$\bar D_4$ has three nodes of in-degree $0$, namely $v_1$, $v_2$, and $v_4$.
The two child edges of $v_1$ are represented by the red dashed parent edges of
$H(v_3)$ and $H(d)$ in $\tilde H_4$.
According to the wiring guesses for $v_3$ and~$d$, their top endpoints are
red.
Since $v_1$ belongs to the red dashed tree, $v_1$ is free.
The same two edges also represent the child edges of $v_2$.
However, $v_2$ is green and, thus, is not free.
Finally, the two child edges of $v_4$ are represented by the blue solid parent
edges of $H(c)$ and $H(d)$, both of which have blue top endpoints according
to the wiring guesses for $c$ and $d$.
Thus, $v_4$ is also free.
Assume we choose $v_4$ as the next node to add to $D_4$ to obtain $D_5$.
We create the node $H(v_4)$ in $\tilde H_5$ by identifying the top endpoints
of the blue solid parent edges of $H(c)$ and $H(d)$ and then create a parent edge
of $H(v_4)$ according to $v_4$'s wiring guess.
This produces the graph $\tilde H_5$ in
Figure~\ref{fig:reconstruction-signature}.
Since the colour sets of the two child edges of $H(v_4)$ have only the blue
solid tree in common, $v_4$ has no buddies and $D_5$ has node set
$\set{b, c, d, v_3, v_4}$.

Nodes $v_1$ and $v_2$ are the only nodes of in-degree $0$ in $\bar D_5$.
Just as in $\bar D_4$, $v_1$ is free and $v_2$ is not.
Thus, we choose $v_1$ as the node to add to $D_5$ to obtain $D_6$.
The two child edges of $v_1$ are represented by the red dashed parent edges
of $H(v_3)$ and $H(d)$ in $\tilde H_5$.
We identify their top endpoints to create the node $H(v_1)$ and add parent
edges according to the wiring guess for $v_1$.
This produces the graph $\tilde H_6$ in
Figure~\ref{fig:reconstruction-signature}.
Now observe that the two child edges of $H(v_1)$ in $\tilde H_6$ are also
coloured green (dotted).
Thus, we identify the node that is the parent of the two edges of the green
dotted tree represented by these child edges, which is node $v_2$.
Node $v_2$ becomes a buddy of $v_1$ and is added to $D_5$ along with $v_1$
to obtain $D_6$.
Thus, $D_6$ has the node set $\set{b, c, d, v_1, v_2, v_3, v_4}$.
Note that making $v_1$ and $v_2$ buddies does not create any conflicts because
they both have the same wiring guess in $\G$.

Finally, the only node remaining in $\bar D_6$ is $\set{a, e, \rho}$.
The root edges of $\tilde H_6$ represent exactly the set of pendant edges
of the AAF component $\set{a, e, \rho}$ in the three input trees, so
$\set{a, e, \rho}$ is free in $\bar D_6$.
We create a node $H(\set{a, e, \rho})$ in $\tilde H_7$ by identifying the top
endpoints of all root edges of $\tilde H_6$.
Since there is no wiring guess for $\set{a, e, \rho}$ in $\G$, we do not add
any parent edges to $H(\set{a, e, \rho})$, and $\tilde H_7 = \tilde H$ is the
final signature.

It is easily verified that we would have obtained the exact same signature had
we chosen different nodes to add to $D_i$ in iterations where $\bar D_i$
contained more than one free node.

\paragraph{Lemma~\ref{lem:network-from-description}: Expanding AAF components.}

In our example, the only non-trivial AAF component to be expanded is the
component $\set{a, e, \rho}$.
This component has two non-root edges.
Let $f_a$ be the parent edge of $a$, and let $f_e$ be the parent edge
of $e$ in this component.
In $\tilde H$, the node $H(\set{a, e, \rho})$ has four child edges:
a red dashed parent edge $e_1$ of $H(b)$, a red dashed parent edge
$e_2$ of $H(v_1) = H(v_2)$,
a green-blue (dotted-solid) parent edge $e_3$ of $H(v_1) = H(v_2)$, and a blue
solid parent edge $e_4$ of $H(v_4)$.
Edges $e_1$ and $e_2$ attach to $f_a$ in the red dashed tree and do not
represent any edges in any other trees, so we add them to $E_{f_a}$.
Edge $e_3$ attaches to $f_e$ in the green dotted and the blue solid trees, so
there is no conflict and we add it to $E_{f_e}$.
Edge $e_4$ attaches to $f_e$ in the blue solid tree and does not represent any
\leo{edge} in any other \leo{tree}, so we add it to $E_{f_e}$.

The DAG $D_{f_a}$ has two nodes representing edges $e_1$ and $e_2$ with
an edge from $e_2$ to $e_1$ because $e_2$ attaches to $f_a$ above $e_1$
in the red dashed tree.
A topological ordering of $D_{f_a}$ places these edges in the order
$\seq{e_2, e_1}$, and \leo{this} is the order in which we attach $e_2$ and $e_1$
to $f_a$.

The DAG $D_{f_e}$ has two nodes representing edges $e_3$ and $e_4$ with
an edge from $e_3$ to $e_4$ because $e_3$ attaches to $f_e$ above $e_4$
in the blue solid tree.
The green dotted tree does not impose any conflicting ordering constraints
because only edge $e_3$ belongs to this tree.
A topological ordering of $D_{f_e}$ places $e_3$ and $e_4$ in the order
$\seq{e_3, e_4}$, and this is the order in which we attach $e_3$ and
$e_4$ to $f_e$.
The result is the \revA{TNET} shown in
Figure~\ref{fig:reconstruction-output}.

\leo{Finally, the hybridization network induced by this \revA{TNET} is shown in Figure~\ref{fig:reconstruction-final}.}

\end{document}